\documentclass[11pt]{article}

\usepackage{amsmath,amsfonts,amssymb,graphicx,psfrag,rotating}
\usepackage[colorlinks=true,urlcolor=blue,citecolor=black,linkcolor=blue]{hyperref}

\usepackage{setspace} \usepackage{dcolumn}

\usepackage{natbib} \bibpunct{(}{)}{;}{a}{}{,}

\usepackage{comment} \usepackage{multirow} \usepackage{lscape}
\usepackage{bbm} \usepackage{scalefnt}
\usepackage[autolanguage]{numprint} \usepackage{caption}
\usepackage{booktabs} \usepackage{bm} \usepackage{threeparttable}
\usepackage{rotating} \usepackage{color}

\usepackage[T1]{fontenc}

\usepackage{qz_package}

\newcommand{\revise}[1]{{\textcolor{Black}{#1}}}

\begin{document}

\title{A note on post-treatment selection in studying racial
  discrimination in policing\footnote{Accepted manuscript.}}

\author[1]{Qingyuan Zhao} \affil[1]{Statistical Laboratory, Department
  of Pure Mathematics and Mathematical Statistics, University of
  Cambridge\thanks{qyzhao@statslab.cam.ac.uk}}

\author[2]{Luke J Keele} \affil[2]{Department of Surgery, Perelman
  School of Medicine, University of Pennsylvania\thanks{luke.keele@uphs.upenn.edu}}

\author[3]{Dylan S Small} \affil[3]{Department of Statistics, Wharton
  School, University of Pennsylvania\thanks{dsmall@wharton.upenn.edu}}

\author[4]{Marshall M Joffe} \affil[4]{Department of Biostatistics,
  Epidemiology and Informatics, Perelman School of Medicine,
  University of Pennsylvania\thanks{mjoffe@mail.med.upenn.edu}}

\date{\today}

\maketitle

\begin{abstract}
  We discuss some causal estimands used to study racial
  discrimination in policing. A central challenge is that not all
  police-civilian encounters are recorded in administrative datasets
  and available to researchers. One possible solution is to consider the
  average causal effect of race conditional on the civilian already
  being detained by the police. We find that such an estimand can be
  quite different from the more familiar ones in causal inference and
  needs to be interpreted with caution. We propose using an
  estimand new for this context---the causal risk ratio, which has more
  transparent interpretation and requires weaker identification
  assumptions. We demonstrate this through a reanalysis of the NYPD
  Stop-and-Frisk dataset. Our reanalysis shows that the naive
  estimator that ignores the post-treatment selection in
  administrative records may severely underestimate the disparity in
  police violence between minorities and whites in these and similar
  data.
\end{abstract}

\vfil

\subsection*{Acknowledgement}

The authors declare no ethical issues or conflicts of interest in this
research. The authors thank Dean Knox, Joshua Loftus, Jonathan
Mummolo, and four anonymous reviewers for their helpful
suggestions. Research data that support the findings of this study are
openly available in the APSR Dataverse at
\href{https://doi.org/10.7910/DVN/ZQMYII}{DOI:10.7910/DVN/ZQMYII}.

\clearpage

\doublespacing

\section*{Introduction}

Evidence of racial disparities in policing is an urgent and highly
relevant policy question in empirical research.
\revise{A growing number of studies have focused on this critical topic
\citep{eckhouse2017descriptive,edwards2019risk,christiani2021better,
baumgartner2018suspect,
shoub2020race,
epp2020use}.} However, studies
  of racial disparities are fraught with methodological challenges
  \citep{ridgeway2006assessing,ridgeway2009doubly,goel2016precinct}.
  Recent work by \citet[hereafter KLM]{knox2020administrative} provides
  important new results on the difficulties of learning about racial disparities in
  policing from administrative data. One key point made by KLM is that
such investigations have an intrinsic selection bias, because
administrative records only contain those encounters in which
civilians are detained. If there is racial discrimination in police
detainment in the first place, any naive analysis using the
administrative data may then suffer from potentially severe selection
bias.

Here, we present a research note on this important topic with two
  purposes. First, KLM focused on several local causal
  estimands that are being used in the empirical studies. We demonstrate
  that these local estimands---even when identified with observational
  data---cannot be used to make inferences about
  more global effects like the average treatment effect.
  Second, we introduce a global causal risk ratio estimand that is
  straightforward to interpret and requires fewer assumptions
  to identify than either the local effects considered by KLM or global
  risk differences. Although it still depends on some quantities that
  need to be estimated from external data, we demonstrate how we can use
  Bayes' formula to avoid the hard problem of estimating
  the probability of detainment in police-civilian encounters.
  We conclude this research note with a reanalysis of the
  New York City Police Department (NYPD) Stop-and-Frisk dataset
  and some further discussion. Our empirical results show that a
  naive analysis of police administrative datasets that ignores the
  selection bias can severely underestimate the risk of police force
  for minorities. We present results that suggest a naive approach may
  understate the effect of civilian race on risk of police violence by a
  factor of 10 or more.

\section*{Review}

We begin with a brief review of the key quantities in KLM. Following their work, the unit of analysis is an encounter between civilians and police, \revise{where an encounter is defined as all events in which the police sight a civilian, including those in which a civilian is allowed to pass undisturbed}.
There are $n$ encounters indexed by $i = 1,\ldots,n$.  We denote the outcome with $Y_i$, where $Y_i=1$ indicates the use of force by the
police in encounter $i$. Next, $D_i$ is a binary variable where $D_i=1$ records the race of the civilian as a minority. While the
race of the civilian is not manipulable, we adopt the approach in KLM where the counterfactual is the replacement of the civilian in an
encounter with a separate, comparable civilian engaged in comparable behavior, but differing on race \citep[p.\ 621]{knox2020administrative}.
We use $M_i$ to indicate a police detainment or stop of a civilian. Critically, $M_i=1$ for the subset of encounters that resulted in a stop by the police and are present in the administrative data. Finally, $X_i$ represents a collection of covariates that describe aspects of the stops in the data. These could include
measures for time of day, location, age, sex, and civilian behavior at
the time when first encountered by police. Unless stated otherwise,
conditioning on $X$ is implicit.

For formal causal inference, we introduce the potential outcomes for
$M_i$ and $Y_i$. We have the potential mediator $M_i(d)$ which
represents whether encounter $i$ would have resulted in
a stop if civilian race is $d$. Next, $Y_i(d,m)$ is the potential
outcome for the use of force if race is $d$ and the mediating variable
is set to $m$; similarly, $Y_i(d)$ is the potential outcome if race is
$d$. Throughout this note we make the stable unit treatment assumption
(SUTVA), so $M_i(D_i) = M_i$ and $Y_i(D_i, M_i) = Y_i(D_i) =
Y_i$. This assumption means that the observed mediator (detainment)
  and outcome (use of force) are consistent with their corresponding
  counterfactual values. Hereafter, we assume the variables
  $D_i,M_i,Y_i$ and the potential outcomes of $M_i$ and $Y_i$ are drawn
  independently from the same unknown distribution. To simplify the
  exposition, we will drop the $i$ subscript.

KLM studied the following ``naive'' treatment effect estimand:
\begin{equation}
  \label{eq:naive-ate}
  \Delta = \E[Y \mid D = 1, M = 1] - \E[Y \mid D = 0, M = 1],
\end{equation}
where $\E$ denotes expectation over a random police-civilian
encounter. Intuitively, $\Delta$ compares the average rates
of force between different racial groups who are detained by
police. KLM showed that, if there is racial discrimination in detainment and an
unmeasured confounder between detainment and use of force (see
\Cref{fig:dag}), the naive treatment effect $\Delta$ can be quite misleading when used to
represent the \emph{causal} effect of race on police violence.

\begin{figure}[t] \centering
  % \begin{subfigure}[b]{\textwidth} \centering
  \begin{tikzpicture}[minimum size = 0.3cm]
    \node[circle,draw,name=d]{$D$}; \node[circle,draw,name=m,right= of
    d]{$M$}; \node[circle,dashed,draw,name=u,above right= of m]{$U$};
    \node[circle,draw,name=y,below right= of u]{$Y$}; \draw[->] (d) to
    (m); \draw[->,out=-30,in=-150] (d) to (y); \draw[->] (m) to (y);
    \draw[->,dashed] (u) to (m); \draw[->,dashed] (u) to (y);
  \end{tikzpicture}

  \caption{KLM's directed acyclic graph (DAG) model for racially
    discrimination in policing with an unmeasured mediator-outcome confounder
    $U$. The treatment $D$ is race of the civilian. The
      mediator $M$ is an indicator for police detainment and the
      outcome $Y$ is an indicator for police use of force.
    Administrative records only contain observations with $M=1$.}
  \label{fig:dag}
\end{figure}
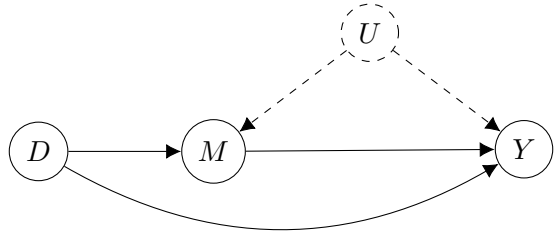

The key issue is that the structure of the data implies all estimates
are conditional on $M$---a post-treatment variable, which often leads to biased
estimators of the causal effect \citep{rosenbaum1984consequences}. Bias of
this type occurs in many applied problems in social science
\citep{elwert2014endogenous,montgomery2018conditioning} and medicine
\citep{paternoster2017genetic}.

Using the principal stratification framework of
\citet{frangakis2002principal}, KLM showed that it is still possible to
either identify or partially identify
certain forms of average treatment effects using a set of tailored
causal assumptions. These assumptions include mandatory reporting,
mediator monotonicity, and treatment ignorability. Specifically,
KLM derived nonparametric bounds for the average treatment
  effect of race on use of force among those who are detained by the
  police:
\begin{equation*} \text{ATE}_{M=1} = \E[Y(1) - Y(0) \mid M=1].
\end{equation*}
They also derived a point identification formula for the
  average treatment effect among those who are minorities and detained
  by the police:
\begin{equation*} \text{ATT}_{M=1} = \E[Y(1) - Y(0) \mid D=1, M=1].
\end{equation*}
Notice that their results rely on an external estimate of the
proportion of \revise{racially motivated detainments among all reported
minority detainments},
i.e., $\P(M(0) = 0 \mid D = 1, M = 1)$. See KLM (p. 631) for
discussion on estimating this quantity. Moreover, KLM also derived an
identification formula for the average treatment effect $\text{ATE} =
\E[Y(1) - Y(0)]$ given external estimates of \revise{the rate of detainments}
$\P(M = 1 \mid D = d)$ by \revise{race} $d = 0,1$.

The identification results in KLM depend crucially on the following
assumption:
\begin{assumption}[Mandatory reporting] \label{assump:m0y0} (i)
  $Y(0,0) = Y(1,0) = 0$ and (ii) the administrative data contains all
  \revise{detainments/stops of civilians by the police}.
\end{assumption}
The first part of this assumption assumes that there will be no
  police violence if the civilian is not stopped in the first place. The
  second part assumes we observe a sample from the conditional
  distribution of the variables given $M=1$, which is essential for
  statistical inference. We will make \Cref{assump:m0y0} throughout this
  note and further discuss its practical implications before the real
  data analysis.

%%%%%%%%%%% Next Section %%%%%%%%%%%%%%%%%
\section*{Average treatment effects conditional on the mediator}

In many causal analyses, investigators are focused on the sample average treatment effect (ATE),
which is the average difference in potential outcomes averaged over the study population. At times,
researchers define the ATE over specific subpopulations which makes
the ATE more local; e.g., the average treatment effect might be
defined for the subpopulation exposed to the treatment or the average
treatment effect on the treated (ATT). Often the ``global''
ATE is the goal in many studies and is preferred over more local
effects \citep[ch. 2]{Gerber:2012}. For example, IV studies have been
strongly critiqued for identifying a local average treatment effect
(LATE) instead of the global ATE
\citep{deaton2010instruments,swanson2014think}. Moreover, even some
defenders of IV studies view the LATE as a ``second choice'' estimand
compared to the global ATE \citep{imbens2014instrumental}.

As KLM outline, the global ATE has not generally been the target
causal estimand in this literature. Instead, researchers have focused
on $\text{ATE}_{M=1}$ and $\text{ATT}_{M=1}$ which are both
conditional on the mediator $M$. Notice that these estimands are not only
more local than the global ATE but also condition on a post-treatment
quantity. Nonetheless, they are not the first estimands in causal
inference that condition on post-treatment quantities. Other examples
of estimands that condition on post-treatment quantities include the
survivor average treatment effect in \citet{frangakis2002principal}
(though conceptually the always survivor principal stratum can be
thought as a pretreatment variable), effect modification by a
post-treatment quantity
\citep{Stephens:2013a,ertefaie2018discovering}, and the probability of
causation $\P[Y(0) = 0 \mid D=1, Y = 1]$
\citep{robins1989probability,pearl1999probabilities,dawid2017probability}.

The local effects in this context may have important policy
  relevance. As such, the preference for a global ATE may not always be warranted
  in this domain. However, \revise{an} inexperienced researcher might think these
estimands are informative about the
global ATE or even an estimand such as the controlled direct effect:
$\E[Y(1,1) - Y(0, 1)]$.  Next, we build upon the population
stratification framework in KLM and clarify the difference between the
conditional estimands in KLM and estimands like the global ATE.

To simplify the illustration, we will consider the case
where there is no mediator-outcome confounder (i.e.\ no variable $U$
in the diagram in \Cref{fig:dag}). The issues we describe below will still
occur if there is mediator-outcome confounding. In mediation analysis,
a standard way to decompose the average treatment effect is
\[ \text{ATE} = \E[Y(1) - Y(0)] = \E\big[Y(1, M(1)) - Y(1, M(0))\big]
  + \E\big[Y(1, M(0)) - Y(0, M(0))\big].
\] The two terms on the right hand side are called the pure indirect
effect (PIE) and pure direct effect (PDE)
\citep{robins1992identifiability}. Under the Non-Parametric Structural
Equation Model with Independent Errors (NPSEM-IE) model
\citep{pearl2009causality,richardson2013single} and
\Cref{assump:m0y0}, they can be expressed as
\begin{align*} \text{PIE} = \beta_M \cdot \E[Y(1,1)],~\text{PDE} =
  \beta_Y \cdot \E[M(0)],
\end{align*} where $\beta_M = \E[M(1) - M(0)]$ is the average effect
of race on detainment and $\beta_Y = \E[Y(1,1) - Y(0,1)]$ is the
controlled direct effect of race on police violence (See the
Appendix). An immediate consequence of the above expressions
is that
\begin{equation}
  \label{eq:consistent-property} \text{ATE} \ge
  0~\text{if}~\beta_M,\beta_Y \ge 0~\text{and}~\text{ATE} \le
  0~\text{if}~\beta_M,\beta_Y \le 0.
\end{equation}
In words, the global ATE is nonnegative whenever both \revise{the direct and
indirect effects}
are nonnegative, and vice versa. This same property also holds for the
ATT because in the simple
setting here the treatment $D$ is completely randomized.

In the Appendix, we use principal stratification to show
that neither $\text{ATE}_{M=1}$ or $\text{ATT}_{M=1}$ is guaranteed to
inherit the sign of $\beta_M$ and $\beta_Y$ and satisfy the property
in Equation~\eqref{eq:consistent-property}. Specifically, we outline
concrete examples in which:
\begin{enumerate}
\item \revise{The pure direct and indirect effects are both positive},
  but $\text{ATE}_{M=1} < 0$;
\item \revise{The pure direct and indirect effects are both negative},
  but $\text{ATE}_{M=1} > 0$ and $\text{ATT}_{M=1} > 0$.
\end{enumerate}
That is, when there is racial discrimination of the same
direction in both police detainment and the use of force, it is still
possible for $\text{ATE}_{M=1}$ and $\text{ATT}_{M=1}$ to have the
opposite sign. We refer the reader to the Appendix for some concrete
counterexamples and further comments on this phenomenon.

In sum, \revise{the local estimands} $\text{ATE}_{M=1}$ and
$\text{ATT}_{M=1}$ are generally
different from the global estimands that are routinely the target in
causal analyses. As such, we urge applied researchers to use caution
when using these local estimands to infer anything about the global
estimands.

%%%%%%%%%%% Next Section %%%%%%%%%%%%%%%%%

\section*{A new estimator for the causal risk ratio}
\label{sec:causal-risk-ratio}

KLM also derived an identification formula for $\text{ATE}_{M=1}$ using external
estimates of \revise{the rate of detainment} $\P(M = 1 \mid D = d)$
for \revise{race} $d=0,1$. Unfortunately, it is
often difficult to quantify the frequency of stops among all
police-civilian encounters, as noted in their paper. In particular, it can be
difficult to determine the magnitude of $\P(M = 1 \mid D = d)$. Here,
we show that by formulating the estimand on a relative scale, we can
avoid this difficulty and obtain point identification.

More specifically, we consider the following causal risk ratio (CRR) for
covariate level $x$:
\[
  \text{CRR}(x) = \frac{\E[Y(1) \mid X = x]}{\E[Y(0) \mid X = x]}.
\]
When this
term is equal to one the risk of police violence does not vary with
the race of the civilian. When this term is greater than one, the risk
of violence is higher for minorities. \revise{Risk ratios, while not commonly used in political science,
have been used in the literature on policing \citep{eckhouse2017descriptive,edwards2019risk,christiani2021better}.
However, previous use of risk ratios has tended to be descriptive way rather than as causal quantities.
Moreover, risk ratios can be a powerful rhetorical tool for understanding discussing racial disparities.
In the context of police violence, it may be tempting to use the
following ratio to measure racial disparities:}
\[
  \text{Naive risk ratio} = \frac{\E[Y \mid D = 1, M = 1, X = x]}{\E[Y
    \mid D = 0, M = 1, X = x]}.
\]
\revise{This quantity divides the rates of police violence experienced
by minorities and non-minorities, given that they have the same
covariate $x$ and are detained by the police. We will see below that
the naive risk ratio is generally not the same as the causal risk
ratio due to conditioning on the colliding variable $M$ (detainment);
in fact, these two quantities can be drastically different.}

Expressing results in a relative fashion can be an effective way of
communication
especially when the risk of police violence is fairly low among a
specific population. For example, let's say in
one specific locale, the risk of police violence for Black residents
is .01\% and is .001\% for white residents. The difference in these
risks is obviously very small. However, in relative terms, the risk of
police violence is 10 times higher for black residents than for white
residents. As such, even if the absolute risk is low, a large increase
in relative risk is likely to be of significant interest.

Using treatment ignorability (i.e.\ the DAG model in \Cref{fig:dag} conditional on
$X$) and \Cref{assump:m0y0}, the causal effect of race can be
identified based on the decomposition
\[ \E[Y(d) \mid X = x] = \E[Y \mid M = 1, D = d, X = x] \cdot \P(M = 1
  \mid D = d, X = x),~\text{for}~d=0,1.
\]
The same result is derived in KLM and forms the basis of their
identification of the ATE. We simplify their proof in
the Appendix and show that some of KLM's
identification assumptions can be relaxed. Specifically,
we can arrive at the same result without invoking mediator
monotonicity and relative nonseverity of racial stops (Assumptions 2
and 3 in KLM).

By using Bayes formula for the last term on the right hand side (see
the Appendix), we obtain the following identification
result:
\begin{equation}
  \label{eq:rr-identification} \text{CRR}(x) = \underbrace{\frac{\E[Y \mid D = 1, M
      = 1, X = x]}{\E[Y \mid D = 0, M = 1, X = x]}}_{\text{naive risk ratio}}
  \cdot \underbrace{\Big\{ \frac{\P(D = 1 \mid
      M = 1, X = x)}{\P(D = 0 \mid M = 1, X = x)} \Big\} \Big/ \Big\{\frac{\P(D = 1 \mid
      X = x)}{\P(D = 0 \mid X = x)}\Big\}}_{\text{bias factor}}.
\end{equation}
Therefore, by targeting the causal risk ratio, we are able to avoid
the difficulties associated with estimating \revise{the absolute rate of
detainment} $\P(M=1)$ through cancellation.

The first term on the right hand side of \eqref{eq:rr-identification}
is the naive risk ratio estimand conditional on baseline
covariates. It is the risk ratio counterpart to the naive risk difference
in \eqref{eq:naive-ate} and both of them ignore the possible
bias from the selection process into the administrative data. The
second term inside the curly brackets  is a ratio of probability
ratios. The first ratio of
probabilities measures the relative probability of an \revise{detainment}
being with a minority conditional on \revise{covariate} $X=x$, \revise{which can be
  estimated from} the administrative data. The second ratio also measures the relative
probability (odds) of an encounter being with a minority conditional
on \revise{covariate} $X=x$, but these probabilities need to be \revise{approximated or
bounded with} a second data source. This ratio
between the last two terms is thus an odds ratio that characterizes
the bias of the naive estimator; for this reason, we call it the
``bias factor.'' That is, if minorities are over-represented in the
administrative data, the bias factor corrects that over-representation
and so increases the magnitude of the risk ratio. For example, if the
probability of a \revise{detainment} being with a minority is 0.8 in
the administrative data and
0.25 in a random police-civilian encounter, the bias factor would be
$(0.8/0.2) / (0.25/0.75) = 12$, which would increase the magnitude
of the naive risk ratio when it is larger than 1. All the terms in
\eqref{eq:rr-identification} can be estimated using generalized linear
models (such as logistic regression), or one could use more flexible
models. \revise{As we highlight below, the available data sources that
can be used for estimating the third term in \Cref{eq:rr-identification} are
imperfect. We demonstrate how a sensitivity analysis can
be used to probe this deficiency. See the end of next section for an
example.}
Confidence intervals can be estimated using the
bootstrap or the delta method.

Note that if we are willing to assume stochastic mediator monotonicity: $\E[M(1) \mid X =
x] \ge \E[M(0) \mid X = x]$ (that is, there is racial bias against the
minority in detainment), the bias factor can indeed be lower bounded
by $1$. In this case, the naive risk ratio (first term on
the right hand side of \eqref{eq:rr-identification}) provides a lower
bound for the causal risk ratio $\text{CRR}(x)$.

While the risk ratio estimand does avoid Assumptions 2 and 3 in KLM critical complications are still present.
That is, the constraints that tend to arise from the use of two data sources remain
a significant source of complexity. In particular, the administrative
dataset can only be used to estimate the first two
terms on the right hand side of \eqref{eq:rr-identification}.
We must find an additional data source that allows us to estimate \revise{the racial distribution conditional on the covariates---}$\P(D = 1
\mid X = x)$ and $\P(D = 0 \mid X = x)$, since the administrative data only contain those encounters where $M = 1$.
However, secondary data sources \revise{tend to also contain data on stops rather than encounters (sightings of
civilians by the police). As such, typically, we use population level data on police stops to approximate encounter rates by racial group.
To the extent these quantities are proportional, the method will be accurate.
However, to the extent these quantities differ, the measure will be biased.} Moreover, there may be measurement inconsistencies between the secondary data and the administrative data. \revise{This can be partly addressed by a sensitivity
    analysis; see the next section for an example. See also
    \citet{knox2020toward} for further discussion on the usage of
    external datasets in this context.}

\revise{Take the NYPD database of police stops as an example. This data source was used in KLM and will be reanalyzed in the next section.}
For a second data source, we will use the Current Population Survey (CPS),
which contains measures for race and also has geographic information that allows us to restrict the data to
the metro area in the state of New York (which is larger than the five
boroughs of New York City). However, The CPS does not contain any more
fine-grained geographic identifiers \revise{or any measures of police
encounters or stops}. Another data source \revise{we will use} is the
Police-Public Contact Survey (PPCS) collected by the U.S. Department
of Justice. However, PPCS is a national survey and geographic identifiers are not
available to researchers. As such, if we use the PPCS, we can do little to
measure the prevalence of police-minority interactions in New York
City. \revise{Additionally, the PPCS collects data on police stops and
not encounters}. As such, we cannot measure rates of encounters
with either data source.

In other settings such as traffic stops, one may use the ``veil
of darkness'' test \citep{grogger2006testing} and use night-time
police stops in the same dataset to estimate the bias factor, as
police are less likely to know the race of a motorist. However, this
still requires the assumption that the racial distribution of motorists
is the same during the day and at night. \revise{Moreover, data sources on encounters are exceedingly rare, and despite
the limitations}, as we show next,
the results using the risk ratio with different data sources can still be useful and illuminate the probable
bias in the naive estimator. They can also serve as the baseline of a sensitivity analysis.

We conclude this section, with a final comment on data constraints.
Identification of the risk ratio estimand as well as those derived in
KLM depend on mandatory reporting (\Cref{assump:m0y0}). It is
important to note that this assumption is both
a restriction on potential outcomes and a feature of the data collection.
The first part of the assumption says that the potential outcome
$Y(d,m)$ is equal to 0 whenever $m = 0$. This assumption is reasonable
because, besides inadvertent collateral damage, there should be
virtually no police violence if the civilian is not stopped by the
police in the first place. The second part of the assumption is needed
so that we can use the
administrative dataset to get the conditional distribution of
$(D,Y,X)$ given $M = 1$. For a given administrative data source, it is
possible that some police stops are unrecorded. If that is the case,
any analysis relying on \Cref{assump:m0y0} needs to be interpreted
with care. \revise{This is not a major concern in the NYPD dataset
  reanalyzed below, as all NYPD police officers are required to report
  all the stops.}

\section*{A reanalysis of the NYPD Stop-and-Frisk dataset}
\label{sec:reanalysis}

We used the identification formula \eqref{eq:rr-identification} to
estimate the causal risk ratio using the NYPD ``Stop-and-Frisk''
dataset analyzed in \citet{fryer2019empirical} and KLM. Specifically,
we use the replication data from KLM. As such, we followed KLM's
preprocessing of the dataset, with the one exception that we removed all races other than
black and white. We also focused on all forms of force rather than
estimate the effects for different types of force. We used CPS
2013 and PPCS 2011 data to estimate the
third term in
\eqref{eq:rr-identification}. \revise{See the end of this section for
  a sensitivity analysis where we perturb the estimates from census
  data.} Because PPCS does not contain a
geographic identifier, we also used the racial distributions for
different subsets of the PPCS data. Specifically, we used subgroups
for those in the survey that experienced a motor vehicle stop, any
other kind of police stop, and those in a large metro area. We
  further explored weighting the PPCS respondents by their reported
  number of face-to-face contacts with the police. Respondents with
  more than 30 reported contacts with the police were excluded in that
  analysis. See \Cref{sec:survey} for details on the exact survey
items we used in this analysis. \revise{As we noted above, neither CPS or
  PPCSD records police-civilian encounters per our definition
  (sighting of civilians), so they can only be regarded as
  approximations of the actual racial distribution in encounters.}

\Cref{tab:est} reports the estimated risk ratios using different
estimators and external datasets. Using the naive estimator---the
first term in \eqref{eq:rr-identification}, we find a modest causal
effect: black people have 29\% higher risk of the police using of
force than white people. Recall that we can view this as lower bound
on the true causal risk ratio if we are willing to assume stochastic
mediator monotonicity (i.e.\ there is discrimination against black
civilians in police detainments on average). The estimator
\eqref{eq:rr-identification} that adjusts for the selection bias shows
a very different picture. No matter which external dataset we used,
the estimated risk ratio for black versus white is always greater than
10.

\begin{table}[t]
  \centering
  \begin{tabular}{lcc}
    \toprule
    External dataset & Estimated risk ratio & 95\% Confidence interval \\
    \midrule
    \multicolumn{3}{c}{Naive estimator---First term in
    \eqref{eq:rr-identification}} \\
    None & 1.29 & 1.28--1.30 \\
    \midrule
    \multicolumn{3}{c}{Adjusted for selection bias by using
    \eqref{eq:rr-identification}} \\
    CPS & 13.6 & 12.8--14.3 \\
    PPCS & 32.3 & 31.3--33.3 \\
    PPCS (MV Stop) & 29.5 & 26.9--32.7 \\
    PPCS (Stop in Public) & 29.2 & 23.5--36.5\\
    PPCS (Large Metro) & 16.7 & 15.4--18.4\\
    PPCS* & 31.1 & 27.9--34.7 \\
    PPCS* (Large Metro) & 19.9 & 14.2--29.0
    \\
    \bottomrule
  \end{tabular}
  \caption{Estimates of the causal effect of minority race (black) on
    police violence. CPS is the Current Population Survey. PPCS is
    Police-Public Contact Survey. PPCS* is PPCS with the respondents
      weighted by their reported number of face-to-face contacts with
      the police. MV Stop is the subset of survey
    respondents that has been the passenger in a motor vehicle that
    was stopped by the police. Large Metro is the subset that lives in
    a region with more than 1 million population. Confidence intervals
    were computed using the nonparametric bootstrap.}
  \label{tab:est}
\end{table}

The estimates in \Cref{tab:est} did not condition on any
  covariate that confounds the effect of race on police use of
force. In \Cref{sec:strat-analys-age}, we report the results of a
stratified analysis by age and gender of the civilian. The estimates
are broadly consistent with those reported in \Cref{tab:est}, but it
appears that female minorities has a much smaller risk ratio (less
discriminated against) than male minorities. Age does not appear to be
an important effect modifier.

Another potentially important confounder is the location of the
police-civilian encounter. However, detailed geographic information is
not available in CPS or PPCS. The NYPD currently has 77 precincts that
are responsible for the law enforcement within a designated geographic
area. Using census blocks and the 2010 census data,
\citet{keefe2020precinct} constructed a population breakdown for each
NYPD precinct. This allows us to compare the proportion of black
residents (among black and white residents) with the proportion of
detainments of black civilians in each precinct (\Cref{fig:race}). It is evident
from this figure that in most of the precincts, black civilians make
up less than half of the population but more than half of the detainment
records. This shows that the bias factor in
\eqref{eq:rr-identification} can be quite large in this problem.

By using the census data to estimate the last term in
\eqref{eq:rr-identification}, \Cref{fig:precinct-estimator} compares
the naive risk ratio estimator and selection-adjusted risk ratio
estimator for each precinct. The selection-adjusted estimates are
almost always much larger except for three outliers---precincts 67 and
113, where Blacks account for more than 90\% of the population, and
precinct 22 (Central Park), where only 25 residents were
recorded, and the majority of police-civilian encounters were likely
with non-residents. It is likely that in these precincts, the
residential distribution in the census data poorly approximate the
racial distribution in police-civilian encounters, because the
civilians could be visitors from other precincts or anywhere else in the
world. Most of the precincts with the highest
estimated risk ratios are wealthy neighborhoods in Manhattan and Brooklyn. In
several precincts, our method estimated that the risk of police use of
force for Blacks is more than 30 times higher than the risk for
whites. This may be due in part to increased suspicion of minorities
in areas where there presence is not common. Finally,
\Cref{fig:precinct-causal-by-pop-a} shows a strong negative correlation
between the estimated risk ratios and the percentage of black
residents in the precinct. This indicates that the racial
discrimination in police use of force may be strongly moderated by
characteristics of the geographic location such as the racial
composition, affluence, and average crime rate of the neighborhood.

\revise{The above analysis relies on the assumption that the racial
  distribution in police-civilian encounters can be well approximated
  by the racial distribution in census or survey datasets. A sensitivity
  analysis can be useful to gauge the potential bias due to poor
approximations of the racial distribution in police-civilian encounters.}
\Cref{fig:precinct-causal-by-pop-b} presents such a sensitivity analysis,
  in which \revise{the civilians encountered with police are
  assumed to be} a mixture of local and city-wide
  residents. More precisely, this sensitivity analysis assumes that in
  each precinct, there is a
  90\% chance of the police encountering a local resident and a 10\%
  chance of the police encountering a resident from another precinct.
  According to the census data, 36.7\% of the population in New York City (excluding
  races other than black and white) was black in 2010. Thus, in this
  sensitivity analysis, the presumed proportion of encounters with black
  civilians is higher than the proportion of black residents in the
  precinct, if the proportion of black residents is lower than
  36.7\%. This shrinks the estimated causal risk ratio towards a common
  value, especially for precincts that are predominantly white or
  predominantly black, as shown in \Cref{fig:precinct-causal-by-pop-b}.

\begin{figure}[t]
  \centering
  \begin{subfigure}[t]{0.6\textwidth}
    \includegraphics[width =
    \textwidth]{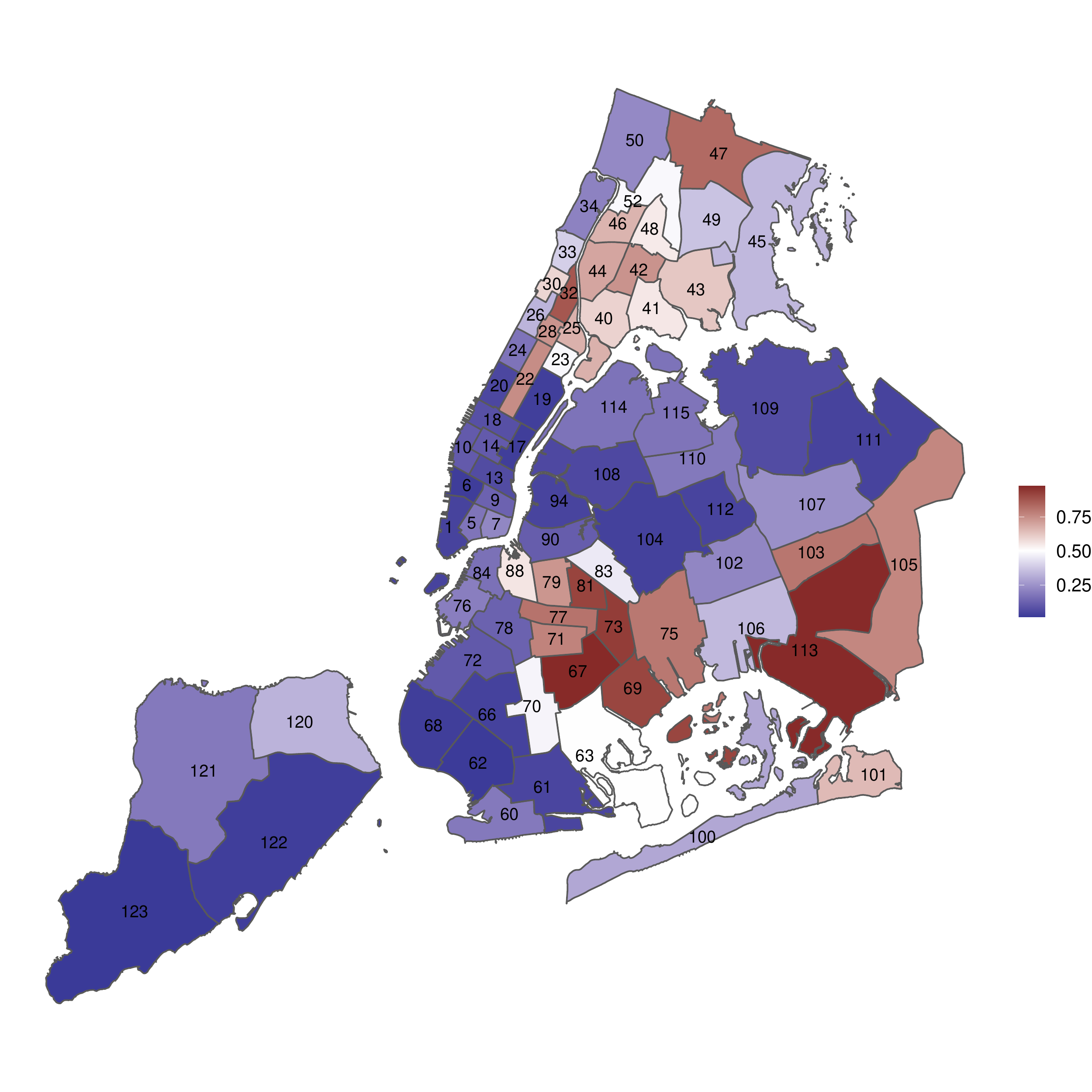}
    \caption{Proportion of black residents in the census data.}
  \end{subfigure} \\
  \begin{subfigure}[t]{0.6\textwidth}
    \includegraphics[width =
    \textwidth]{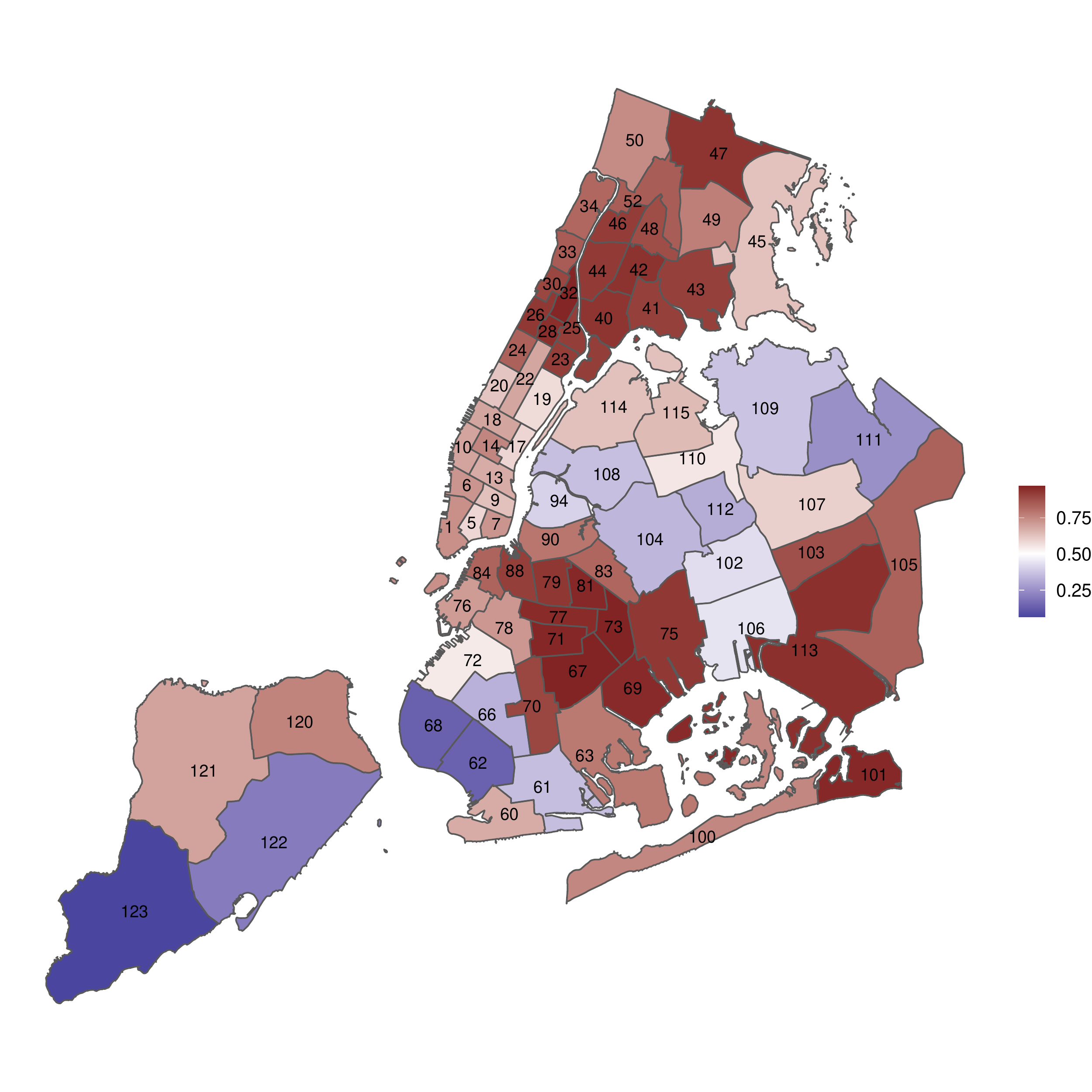}
    \caption{Proportion of detainments of black civilians in the NYPD
      stop-and-frisk data.}
  \end{subfigure}
  \caption{Racial distributions (indicated by
    the filled color) in each NYPD precinct.}
  \label{fig:race}
\end{figure}

\begin{figure}[t]
  \centering
  \includegraphics[height = 0.9\textheight]{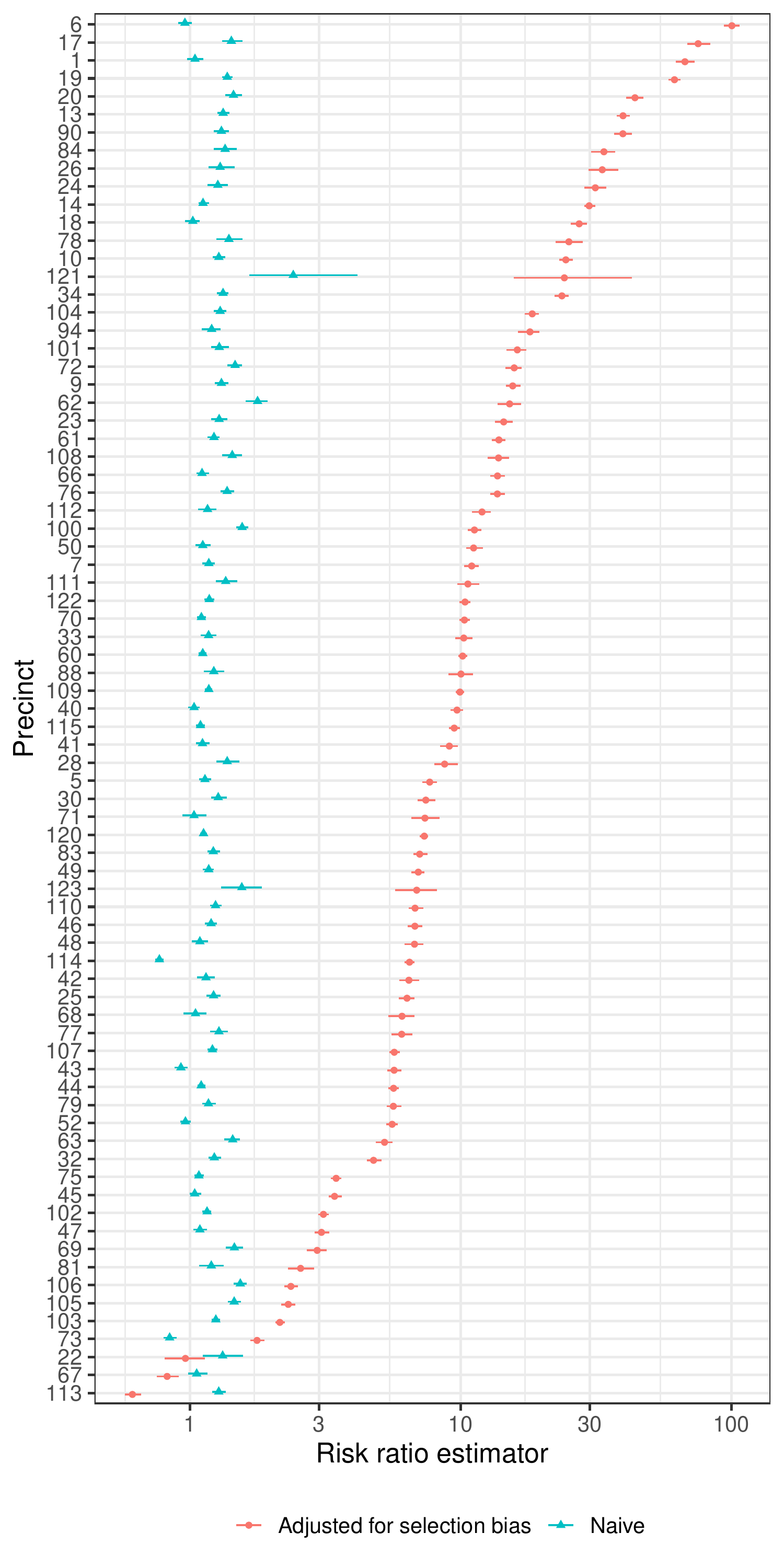}
  \caption{Risk ratio estimates for every NYPD precinct. Error bars
    correspond to 95\% confidence intervals computed by the
    bootstrap. We did not resample the census data because that is
    already the residential distribution (instead of a statistical
    estimate). Blue estimates are obtained using the naive estimator
    (first term in \eqref{eq:rr-identification}); Red estimates
    further take into account the bias factor due to sample selection
    in \eqref{eq:rr-identification}.}
  \label{fig:precinct-estimator}
\end{figure}

\begin{figure}[t]
  \centering

  \begin{subfigure}[t]{0.49\linewidth}
    \includegraphics[width
    =\textwidth]{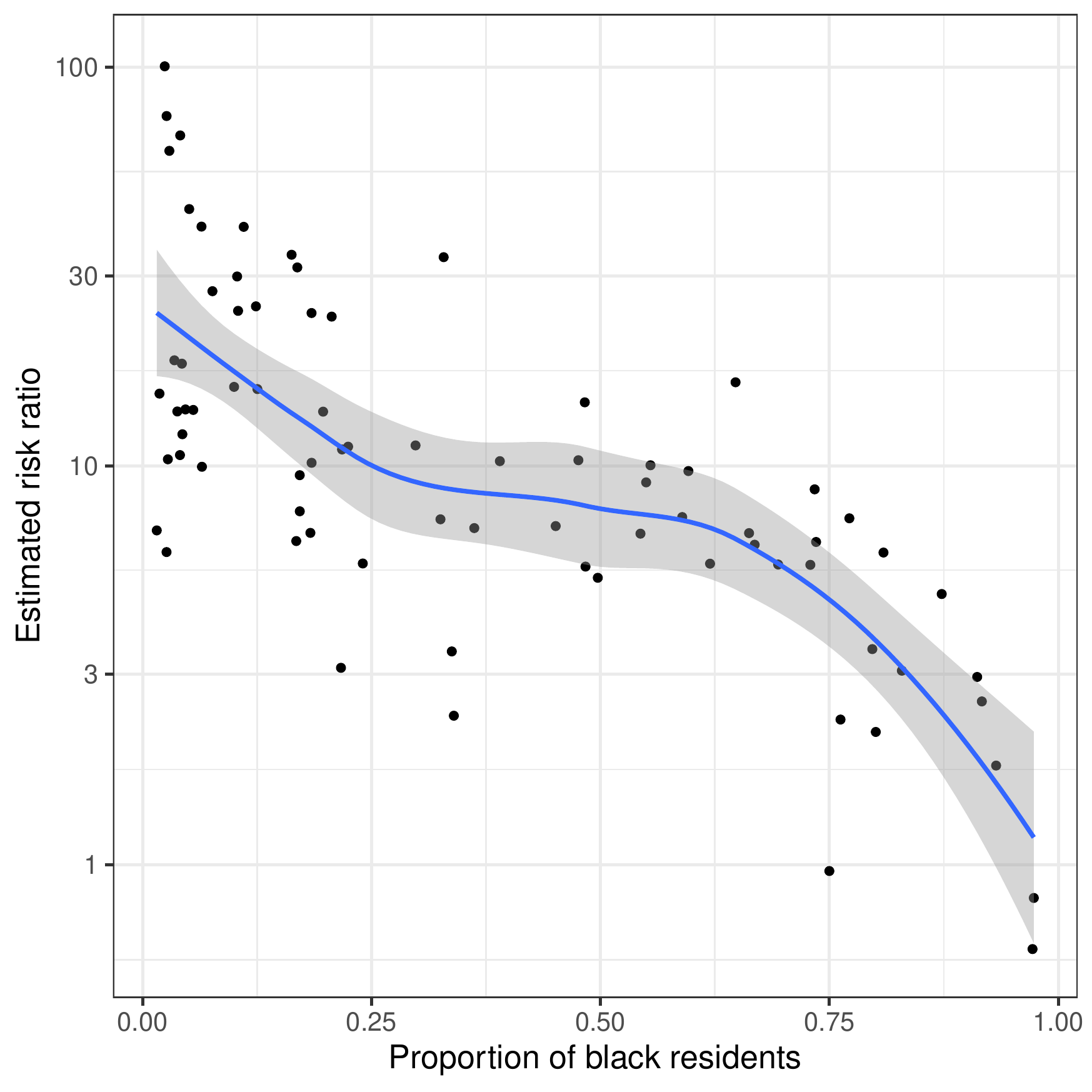}
    \caption{Estimated risk ratio versus proportion of black residents
      in each precinct.}
    \label{fig:precinct-causal-by-pop-a}
  \end{subfigure}
  \begin{subfigure}[t]{0.49\linewidth}
    \includegraphics[width
    =\textwidth]{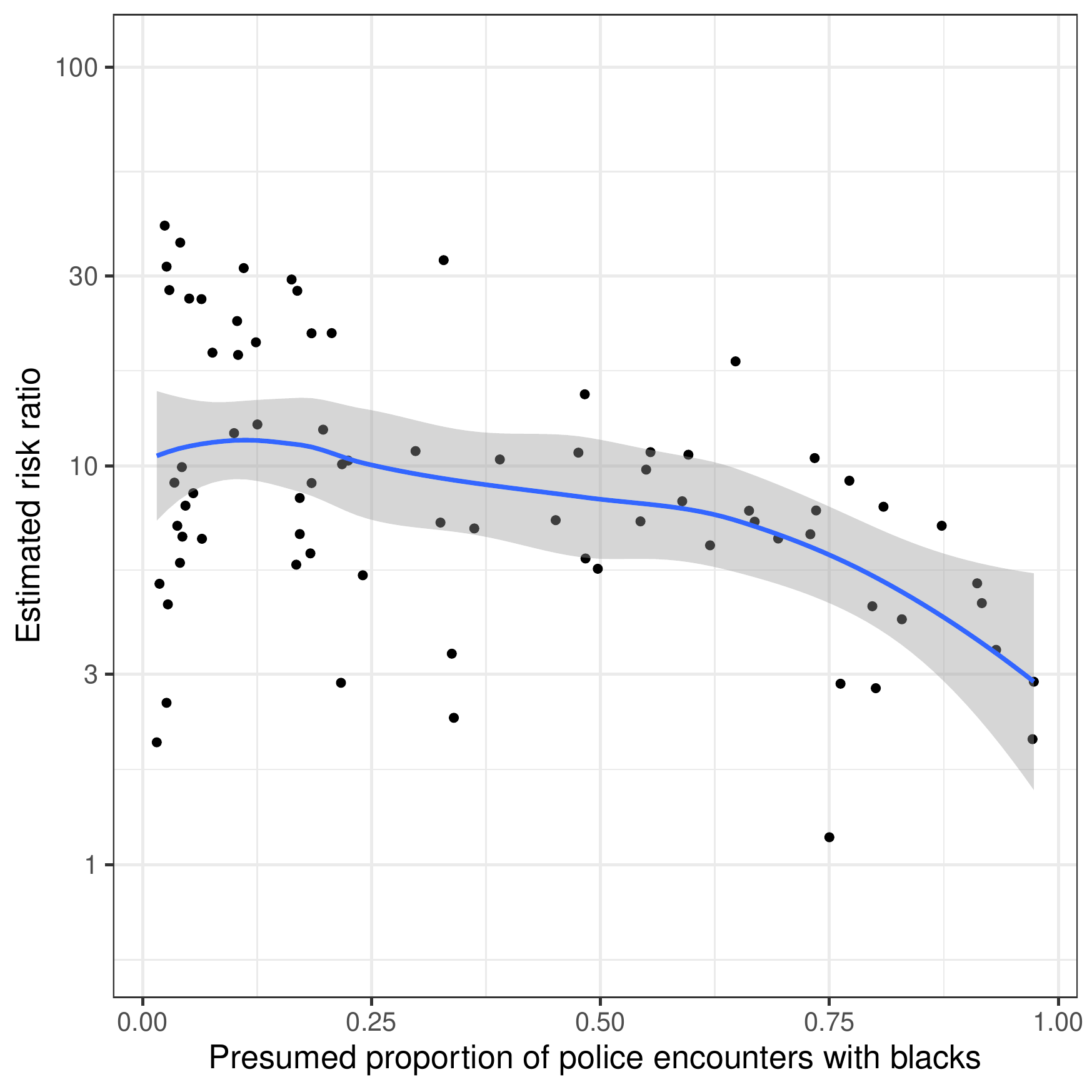}
    \caption{Estimated risk ratio in a sensitivity analysis
        versus proportion of black residents. In each precinct, we
        assume the police encounters a mixture of 90\% local residents
        in the precinct and 10\% city-wide residents.}
    \label{fig:precinct-causal-by-pop-b}
  \end{subfigure}
  \caption{Relationship between the risk ratio estimates (with
    confidence interval error bars) and the
    proportion of black residents across NYPD precincts. The blue
    curves and the shaded confidence regions were obtained by the function
    \texttt{loess} in \textsf{R} with the default span. Notice that
    to use the identification formula \eqref{eq:rr-identification} for
    the risk ratio correctly, we need to estimate the racial
    distribution in police-civilian encounters using an external
    dataset. The residential distribution is used as an approximation
    in \Cref{fig:precinct-causal-by-pop-a}, but it could be biased and
    exaggerate the effect modification as shown in the sensitivity
    analysis in \Cref{fig:precinct-causal-by-pop-b}. See
    Conclusions for further discussion.}
  \label{fig:precinct-causal-by-pop}
\end{figure}

\section*{Conclusions}
\label{sec:conclusions}

% \qz{I will write this in the end.}

In this research note, we studied some causal estimands in the context
of racial discrimination in policing. We found that the ATE that
\revise{conditions} on the mediator (police detainment) can differ in
  sign from the unconditional ATE and other routinely used causal
estimands, so extra caution is needed when using these estimands and
interpreting the results. We also proposed a new estimator for
the causal risk ratio, which is straightforward to interpret and avoids
the difficult task of discerning the percentage of stops in all
police-civilian encounters. In a reanalysis of the
NYPD Stop-and-Frisk dataset with causal risk ratio being the estimand,
we found that for blacks the risk of experiencing force is much higher
than for whites.

When interpreting the results of our reanalysis, the reader should
keep in mind its limitations. First, it is difficult to
find a good external dataset to estimate the bias factor. The datasets
we used should only be viewed as crude approximations to the racial
distribution in police-civilian encounters in New York City. Second, our
  measure of the causal risk ratio is conditional on covariates $X$;
  identification requires treatment ignorability conditional on
  confounders included in $X$. In principle, that would involve
  conditioning simultaneously on confounders like time, location, and
  other relevant characteristics of the police-civilian
  encounter. However, such covariates are not always available
  in external datasets and our analysis only conditions on NYPD
  precinct. Additionally, our method does not yet have a way to
  summarize over multiple covariate strata even if the conditional
  risk ratios are identified and estimated. \revise{Since we did not
    use visible features of the civilians that are associated with
  race and criminal activity (they are not available in the data)},
this may have led to overestimation of the effect of race on use of
force. It is highly implausible that this
  bias could fully explain the large measures of association found
  here. Finally, since New York
is a metropolitan in which
people move around a great deal on a daily basis,
the racial distribution of the residents in a precinct
might \revise{poorly} represents the racial distribution in
police-civilian encounters, especially when the residential
distribution is extreme as demonstrated in our sensitivity
  analysis. In other words, \Cref{fig:precinct-causal-by-pop-a}
may have exaggerated the effect modification by the racial
distribution of the local residents. A further analysis on carefully
selected precincts (e.g.\ residential areas with different racial
composition) is needed to better quantify the effect modification.

Nevertheless, our empirical results show that a naive analysis of
police administrative datasets that ignores the selection bias can
severely underestimate the risk of police force for
minorities. \revise{This also highlights the importance of defining
  the causal estimand clearly in observational studies.}
Further careful analyses are needed to better quantify the
racial discrimination in policing and understand the socioeconomic
factors that moderate racial discrimination.

Finally, we offer a concrete suggestion for applied analysts
  based on our results. KLM conclude by outlining a feasible
  research design for policing studies. Our risk-ratio based analysis and the
  associated sensitivity analysis are useful additions to their suggested
  research plan. Our methods provide useful complements to the analyses outlined
  by KLM. Any policing study will depend on strong assumptions and a broad
  set of results that agree will provide higher quality evidence.

\clearpage

\bibliographystyle{apsr_update}
\bibliography{police_ref}

\clearpage
\appendix

\renewcommand\thefigure{\thesection.\arabic{figure}}
\setcounter{figure}{0}

Below is the Online Supplement for ``A note on post-treatment
selection in studying racial discrimination in policing''.

\section{Average treatment effects conditional on the mediator}
\label{sec:app-a}

We assume the variables $(D,M,Y)$ are generated from a nonparametric
structural equation model: $D = f_D(\epsilon_D), M =
f_M(D,\epsilon_M), Y = f_Y(D,M,\epsilon_Y)$ where
$\epsilon_D,\epsilon_M,\epsilon_Y$ are mutually independent
\citep{pearl2009causality}. Potential outcomes for $M$ and $Y$ can be
defined by replacing random variables in the functions by fixed
values; for example, $M(d) = f_M(d,\epsilon_M),~d=0,1$. Because the
errors are independent, $D$, $\{M(0), M(1)\}$, and $\{Y(0,0), Y(0,1),
Y(1,0), Y(1,1)\}$ are mutually independent
\citep{richardson2013single}. We also make the mandatory assumption (\Cref{assump:m0y0}). The derivations below do not
need mediator monotonicity ($M(1) \ge M(0)$).

We next derive expressions of $\text{ATE}_{M=1}$ and
$\text{ATT}_{M=1}$ using two basic causal effects: $\beta_M = \E[M(1)
- M(0)]$, the racial bias in detainment, and $\beta_Y = \E[Y(1,1) -
Y(0,1)]$, the controlled direct effect of race on police violence. To
simplify the interpretation, we introduce a new variable to denote the
the principal stratum (see Figure 2 in KLM):

\[ S =
  \begin{cases} \text{always stop (al)}, & \text{if}~M(0) = M(1) = 1,
    \\ \text{minority stop (mi)}, & \text{if}~M(0) = 0, M(1) = 1, \\
    \text{majority stop (ma)}, & \text{if}~M(0) = 1, M(1) = 0, \\
    \text{never stop (ne)}, & \text{if}~M(0) = M(1) = 0, \\
  \end{cases}
\]

Let $\mathcal{S} = \{\text{al}, \text{mi}, \text{ma}, \text{ne}\}$ be
all possible values for $S$. Using this notation, we have
\[ \beta_M = \sum_{s \in \mathcal{S}}\E[M(1) - M(0) \mid S = s] \P(S =
  s) = \P(S = \text{mi}) - \P(S = \text{ma}).
\]

By using the independence between $M(d)$ and $Y(d,m)$ and
\Cref{assump:m0y0}, it is easy to show that
\[ \bm \theta =
  \begin{pmatrix} \E[Y(1) - Y(0) \mid S = \text{al}] \\ \E[Y(1) - Y(0)
    \mid S = \text{mi}] \\ \E[Y(1) - Y(0) \mid S = \text{ma}] \\ \E[Y(1) -
    Y(0) \mid S = \text{ne}] \\
  \end{pmatrix} =
  \begin{pmatrix} \E[Y(1,1) - Y(0,1)] \\ \E[Y(1,1) - Y(0,0)] \\
    \E[Y(1,0) - Y(0,1)] \\ \E[Y(1,0) - Y(0,0)] \\
  \end{pmatrix} =
  \begin{pmatrix} \beta_Y \\ \beta_Y + \E[Y(0,1)] \\ - \E[Y(0,1)] \\ 0
    \\
  \end{pmatrix}.
\]

Average treatment effects, whether conditional on $M$ or $D$ or not,
can be written as weighted averages of the entries of $\bm \theta$.

\begin{proposition} \label{prop:weighted-average} Suppose there is no
  unmeasured mediator-outcome confounder (i.e.\ no $U$) in
  \Cref{fig:dag}. Under \Cref{assump:m0y0}, the estimands $\text{ATE}_{M=1}$, $\text{ATT}_{M=1}$,
  $\text{ATE} = \E[Y(1) - Y(0)]$, and $\text{ATT} = \E[Y(1) - Y(0) \mid
  D = 1]$ can be written as weighted averages $(\bm w^T \bm \theta) /
  (\bm w^T \bm 1)$ ($\bm 1$ is the all-ones vector) with weights given
  by, respectively,
  \[ \bm{w}(\text{ATE}_{M=1}) =
    \begin{pmatrix} \P(S = \text{al}) \\ \big[\P(S = \text{ma}) +
      \beta_M\big] \P(D = 1) \\ \P(S = \text{ma}) \P(D = 0) \\ 0 \\
    \end{pmatrix},~\bm{w}(\text{ATT}_{M=1}) =
    \begin{pmatrix} \P(S = \text{al}) \\ \P(S = \text{ma}) + \beta_M
      \\ 0 \\ 0 \\
    \end{pmatrix},
  \] and
  \[ \bm{w}(\text{ATE}) = \bm{w}(\text{ATT}) =
    \begin{pmatrix} \P(S = \text{al}) \\ \P(S = \text{mi}) \\ \P(S =
      \text{ma}) \\ \P(S = \text{ne}) \\
    \end{pmatrix} =
    \begin{pmatrix} \P(S = \text{al}) \\ \P(S = \text{ma}) + \beta_M
      \\ \P(S = \text{ma}) \\ \P(S = \text{ne}) \\
    \end{pmatrix}.
  \]
\end{proposition}

\begin{proof} Let's first consider $\text{ATE}_{M=1}$. By using the
  law of total expectations, we can first decompose it into a weighted
  average of principal stratum effects:
  \begin{align*} \text{ATE}_{M=1} &= \E[Y(1) - Y(0) \mid M = 1] =
                                    \sum_{s \in \mathcal{S}} \E[Y(1) - Y(0) \mid M = 1, S = s] \cdot \P(S
                                    = s \mid M = 1).
  \end{align*} We can simplify the principal stratum effects using
  recursive substitution of the potential outcomes and the assumption
  that $D$, $\{M(0), M(1)\}$, and $\{Y(0,0), Y(0,1), Y(1,0), Y(1,1)\}$
  are mutually independent. For $m_0,m_1 \in \{0,1\}$,
  \begin{align*} &\E[Y(1) - Y(0) \mid M = 1, M(0) = m_0, M(1) = m_1]
    \\ =& \E[Y(1, M(1)) - Y(0, M(0)) \mid M = 1, M(0) = m_0, M(1) = m_1]
    \\ =& \E[Y(1, m_1) - Y(0, m_0) \mid M = 1, M(0) = m_0, M(1) = m_1] \\
    =& \E[Y(1, m_1) - Y(0, m_0) \mid M(0) = m_0, M(1) = m_1] \\ =& \E[Y(1,
                                                                   m_1) - Y(0, m_0)].
  \end{align*} The third equality uses the fact that $M \independent
  \{Y(1,m_1), Y(0,m_0)\} \mid \{M(0), M(1)\}$, because given $\{M(0),
  M(1)\}$ the only random term in $M = D \cdot M(1) + (1-D) \cdot M(0)$
  is $D$. Thus $\text{ATE}_{M=1}$ can be written as
  \[ \text{ATE}_{M=1} =
    \bm\theta^T\bm{w}(\text{ATE}_{M=1}),~\text{where}~\bm{w}(\text{ATE}_{M=1})
    =
    \begin{pmatrix} \P(S = \text{al} \mid M = 1) \\ \P(S = \text{mi}
      \mid M = 1) \\ \P(S = \text{ma} \mid M = 1) \\ \P(S = \text{ne} \mid M
      = 1) \\
    \end{pmatrix}.
  \] Similarly, $\text{ATT}_{M=1}$, $\text{ATE}$, and $\text{ATT}$ can
  also be written as weighted averages of the entries of $\bm \theta$,
  where the weights are
  \[ \bm{w}(\text{ATT}_{M=1}) =
    \begin{pmatrix} \P(S = \text{al} \mid D = 1, M = 1) \\ \P(S =
      \text{mi} \mid D = 1, M = 1) \\ \P(S = \text{ma} \mid D = 1, M = 1) \\
      \P(S = \text{ne} \mid D = 1, M = 1) \\
    \end{pmatrix},~ \bm{w}(\text{ATE}) = \bm{w}(\text{ATT}) =
    \begin{pmatrix} \P(S = \text{al}) \\ \P(S = \text{mi}) \\ \P(S =
      \text{ma}) \\ \P(S = \text{ne}) \\
    \end{pmatrix}.
  \]

  Next we compute the conditional probabilities for the principal
  strata in $\bm{w}(\text{ATE}_{M=1})$ and
  $\bm{w}(\text{ATT}_{M=1})$. By using Bayes'' formula, for any $m_0,m_1
  \in \{0,1\}$,
  \begin{align*} &\P(M(0) = m_0, M(1) = m_1 \mid M = 1) \\ \propto&
                                                                    \P(M(0) = m_0, M(1) = m_1) \cdot \P(M = 1 \mid M(0) = m_0, M(1) = m_1)
    \\ =& \P(M(0) = m_0, M(1) = m_1) \cdot \sum_{d=0}^1 \P(M = 1, D = d
          \mid M(0) = m_0, M(1) = m_1) \\ =& \P(M(0) = m_0, M(1) = m_1) \cdot
                                             \sum_{d=0}^1 1_{\{m_d = 1\}} \P(D = d \mid M(0) = m_0, M(1) = m_1) \\
    =& \P(M(0) = m_0, M(1) = m_1) \cdot \sum_{d=0}^1 1_{\{m_d = 1\}} \P(D
       = d).
  \end{align*} The last two equalities used $M=M(D)$ and $D
  \independent \{M(0), M(1)\}$. For this, it is straightforward to
  obtain the form of $\bm{w}(\text{ATE}_{M=1})$ in
  \Cref{prop:weighted-average}. Similarly,
  \begin{align*} \P(M(0) = m_0, M(1) = m_1 \mid D = 1, M = 1) \propto
    \P(M(0) = m_0, M(1) = m_1) \cdot 1_{\{m_1 = 1\}}.
  \end{align*} From this we can derive the form of
  $\bm{w}(\text{ATT}_{M=1})$ in \Cref{prop:weighted-average}.
\end{proof}

\begin{proposition} \label{prop:pie-pde} Under the same assumptions as
  above, $\text{PIE} = \beta_M \cdot \E[Y(1,1)]$ and $\text{PDE} =
  \beta_Y \cdot \E[M(0)]$.
\end{proposition}
\begin{proof} This follows from the definition of pure direct and
  indirect effects and the following identity,
  \begin{align*} \E\big[Y(d, M(d'''))\big] = \E\big[Y(d, 1) \mid M(d') =
    1\big] \cdot \P(M(d') = 1) = \E\big[Y(d, 1)\big] \cdot \P(M(d') = 1),
  \end{align*} for any $d,d' \in \{0,1\}$.
\end{proof}

Using the forms of weighted averages in \Cref{prop:weighted-average},
we can make the following observation on the sign of the causal
estimands when $\beta_M$ and $\beta_Y$ are both nonnegative or both
nonpositive:

\begin{corollary} \label{cor:paradox} Let the assumptions in
  \Cref{prop:weighted-average} be given.  If $\beta_M \ge 0$ and $\beta_Y
  \ge 0$, then $\text{ATE} = \text{ATT} \ge 0$. Conversely, if $\beta_M \le 0$
  and $\beta_Y \le 0$, then $\text{ATE} = \text{ATT} \le 0$. However, both
  of these properties are not true for $\text{ATE}_{M=1}$ and the second
  property is not true for $\text{ATT}_{M=1}$.
\end{corollary}

The fact that $\text{ATT}$ and $\text{ATE}$ would have the same sign
as $\beta_M$ when $\beta_M$ and $\beta_Y$ have the same sign follows
immediately from \Cref{prop:pie-pde}. However, this important property
does not hold for $\text{ATE}_{M=1}$ and $\text{ATT}_{M=1}$. Here are
some concrete counterexamples:
\begin{enumerate}
\item When $\beta_M = \beta_Y = 0.01$, $\P(S = \text{al}) = 0.1$,
  $\P(S = \text{ma}) = 0.05$, $\E[Y(0,1)] = 0.1$, and $\P(D = 1) =
  0.01$, we have $\text{ATE}_{M=1} = -0.003884$.
\item When $\beta_M = \beta_Y = -0.01$, $\P(S = \text{al}) = 0.1$,
  $\P(S = \text{ma}) = 0.05$, $\E[Y(0,1)] = 0.1$, and $\P(D = 1) =
  0.99$, we have $\text{ATE}_{M=1} = 0.002514$.
\item When $\beta_M = \beta_Y = -0.01$, $\P(S = \text{al}) = 0.1$,
  $\P(S = \text{ma}) = 0.05$, $\E[Y(0,1)] = 0.1$, and $\P(D = 1) =
  0.01$, we have $\text{ATT}_{M=1} = 0.0026$.
\end{enumerate}

Heuristically, this is due to the fact that all of the causal
estimands above, including $\beta_M$,
$\beta_Y$, $\text{ATE}$, $\text{ATE}_{M=1}$, and $\text{ATT}_{M=1}$,
only measure some weighted average treatment effect for
police detainment and/or use of force. Conditioning on the
post-treatment $M$ may correspond to unintuitive weights. The
possibility that $\text{ATE}_{M=1}$ and $\text{ATE}$ can have
different signs can be understood from the following iterated expectation:
\[
  \text{ATE} = \text{ATE}_{M=1} \P(M=1) + \E[Y(1) - Y(0) \mid M = 0]
  \P(M = 0).
\]
In this decomposition, the second term may be nonzero and have the
opposite sign of $\text{ATE}_{M=1}$. An inexperienced researcher might
be tempted to drop the second term because of \Cref{assump:m0y0}, as
$Y(0,0) = Y(1,0) = 0$ with probability 1. However, conditioning on
$M=0$ is not the same as the intervention that sets
$M=0$. This means that we cannot deduce $\E[Y(d) \mid M = 0] = 0$ from
$Y(d,0)=0$, because $\E[Y(d) \mid M = 0] = \E[Y(d, M(d)) \mid M = 0]$
is not necessarily equal to $\E[Y(d, 0) \mid M = 0]$.

The fundamental problem driving this paradox is that conditioning on
the post-treatment variable $M$
alters the weights on the principal strata, as shown in
\Cref{prop:weighted-average}. $\text{ATE}_{M=1}$ and
$\text{ATT}_{M=1}$ then depend on not only the racial bias in
detainment and use of force (captured by $\beta_M$ and $\beta_Y$) but
also the baseline rate of violence $\E[Y(0,1)]$ and the composition of
race $\P(D = 1)$. For instance, in the first counterexample above,
even though the minority group $D=1$ is discriminated against in both
detainment and use of force, because the baseline violence is high and
the minority group is extremely small, $\text{ATE}_{M=1}$ becomes
mostly determined by the smaller bias (captured by $\P(S = \text{ma})
= \P(M(0) = 1, M(1) = 0)$) experienced by the much larger majority
group.

We make some further comments on the above paradox. First of all, the
second counterexample can be eliminated if we additionally assume
$\P(D = 1) < 0.5$, that is $D=1$ indeed represents the minority
group. With this benign assumption, one can show that
$\text{ATE}_{M=1} < 0$ whenever $\beta_M, \beta_Y < 0$. Furthermore,
it can be shown that $\text{ATT}_{M=1} < 0$ whenever $\beta_M,\beta_Y
> 0$. So in a very rough sense we might say that as causal estimands,
$\text{ATE}_{M=1}$ is unfavorable for the minority group (because
$\text{ATE}_{M=1}$ can be negative even if both $\beta_M,\beta_Y > 0$)
and $\text{ATT}_{M=1}$ is unfavorable for the majority group (because
$\text{ATT}_{M=1}$ can be positive even if both $\beta_M,\beta_Y <
0$).

Our second comment is about the first counterexample. We can eliminate
such possibility by assuming mediator monotonicity $\P(S = \text{ma})
= 0$, or in other words, by assuming that the majority race group is
never discriminated against in any police-civilian encounter. KLM
indeed used mediator monotonicity to obtain bounds on
$\text{ATE}_{M=1}$ and $\text{ATT}_{M=1}$. So a supporter of the
estimand $\text{ATE}_{M=1}$ may argue that if one is willing to assume
mediator monotonicity, there is no paradox regarding
$\text{ATE}_{M=1}$. However, it is worthwhile to point out that under
mediator monotonicity, the pure indirect effect is guaranteed to be
nonnegative because $\beta_M = \P(S = \text{mi}) - \P(S = \text{ma}) =
\P(S = \text{mi}) \ge 0$. Empirical researchers should be mindful of
and clearly communicate the consequences of the mediator monotonicity
assumption unless it is compelling in the specific application. See
KLM's discussion after their Assumption 2 on when mediator
ignorability may be violated. This concern can be alleviated if future
work can incorporate non-zero $\P(S = \text{ma})$ as sensitivity
parameters in KLM's bounds.

\section{Derivation of the causal risk ratio}
\label{sec:app-risk-ratio}

To simplify the derivation, we will omit the conditioning on $X = x$
below. Fix a $d \in \{0, 1\}$. Using \Cref{assump:m0y0}, $\E[Y(d) \mid
M(d) = 0] = \E[Y(d,0) \mid M(d) = 0] = 0$. Therefore
\begin{align*} \E[Y(d)] &= \E[Y(d) \mid M(d) = 1] \cdot \P(M(d) = 1)
  \\ &= \E[Y(d, 1) \mid M(d) = 1] \cdot \P(M(d) = 1) \\ &= \E[Y(d, 1)
                                                          \mid M(d) = 1, D = d] \cdot \P(M(d) = 1) \\ &= \E[Y \mid M = 1, D = d] \cdot \P(M(d) =
                                                                                                        1).
\end{align*} The third equality above uses treatment ignorability: $D
\independent Y(d,1) \mid M(d)$ (this follows from the single world
intervention graph corresponding to \Cref{fig:dag}); the last equality
follows from the
consistency (or stable unit value treatment) assumption for potential
outcomes. By further using $D \independent M(d)$,
we have $\P(M(d) = 1) = \P(M(d) = 1 \mid D = d) = \P(M = 1 \mid D =
d)$. Plugging this into the last display equation, we have
\[ \E[Y(d)] = \E[Y \mid M = 1, D = d] \cdot \P(M = 1 \mid D =
  d),~d=0,1.
\] Thus we have recovered KLM's Proposition 2 (point identification of
ATE) without assuming their Assumption 2 (mediator monotonicity) and
Assumption 3 (relative nonseverity of racial stops). To get the causal
risk ratio, we only needs to take a ratio between $\E[Y(1)]$ and
$\E[Y(0)]$ and apply Bayes' formula to cancel $\P(M = 1)$.

\section{Implementation details of the empirical analysis}
\label{sec:survey}

To estimate encounter rates in our empirical analysis
using the PPCS data we used the following three survey questions:

\begin{quote}
  The following are questions about any time in the last 12
  months when police have initiated contact with you. In the last 12 months, have you:

  \begin{description}
  \item[V11] Been stopped by the police while in a public place, but not a moving vehicle? This includes being in a parked vehicle.
  \item[V13] Been stopped by the police while driving a motor vehicle?
  \item[V21] Have you been stopped or approached by the police in the last 12 months for something I haven't mentioned?
  \end{description}
\end{quote}

We created two binary measures as indicators of police encounters.
The first measure (Stop in Public in \Cref{tab:est}) was 1 for being stopped by the police if
the respondent answered Yes to either V11 or V21 and 0 otherwise.
We used V13 as the measure for being stopped in a motor vehicle (MV
Stop in \Cref{tab:est}).

In our alternative analysis (labelled as PPCS* in \Cref{tab:est}), the
stop indicators are weighted by the responses to the following question :
\begin{quote}
  \begin{description}
  \item[V30] Thinking about the times you initiated contact with the police and the
    times they initiated contact with you, how many face-to-face contacts
    did you have with the police during the last 12 months?
  \end{description}
\end{quote}
In that analysis, we excluded outliers with more than 30 reported
contacts with the police.

% Specifically, we weighted the PPCS respondents the two measures of
% police encounters by the reported number of face-to-face contacts with
% the police.

\section{Stratified analysis by age and gender}
\label{sec:strat-analys-age}

Our identification \eqref{eq:rr-identification} of the causal risk
ratio depends on conditioning on all the confounders in $X$. Here we
report the results of an additional analysis where the police-civilian
encounters were stratified by the age and gender of the
civilian. Similarly, the survey respondents were also by their age and
gender. The same analysis that generated \Cref{tab:est} were repeated
for each stratum, and the results are reported in
\Cref{fig:stratified}. It appears that gender is an important effect
modifier but age is not.

\begin{figure}[ht]
  \centering
  \includegraphics[width = 0.9\textwidth]{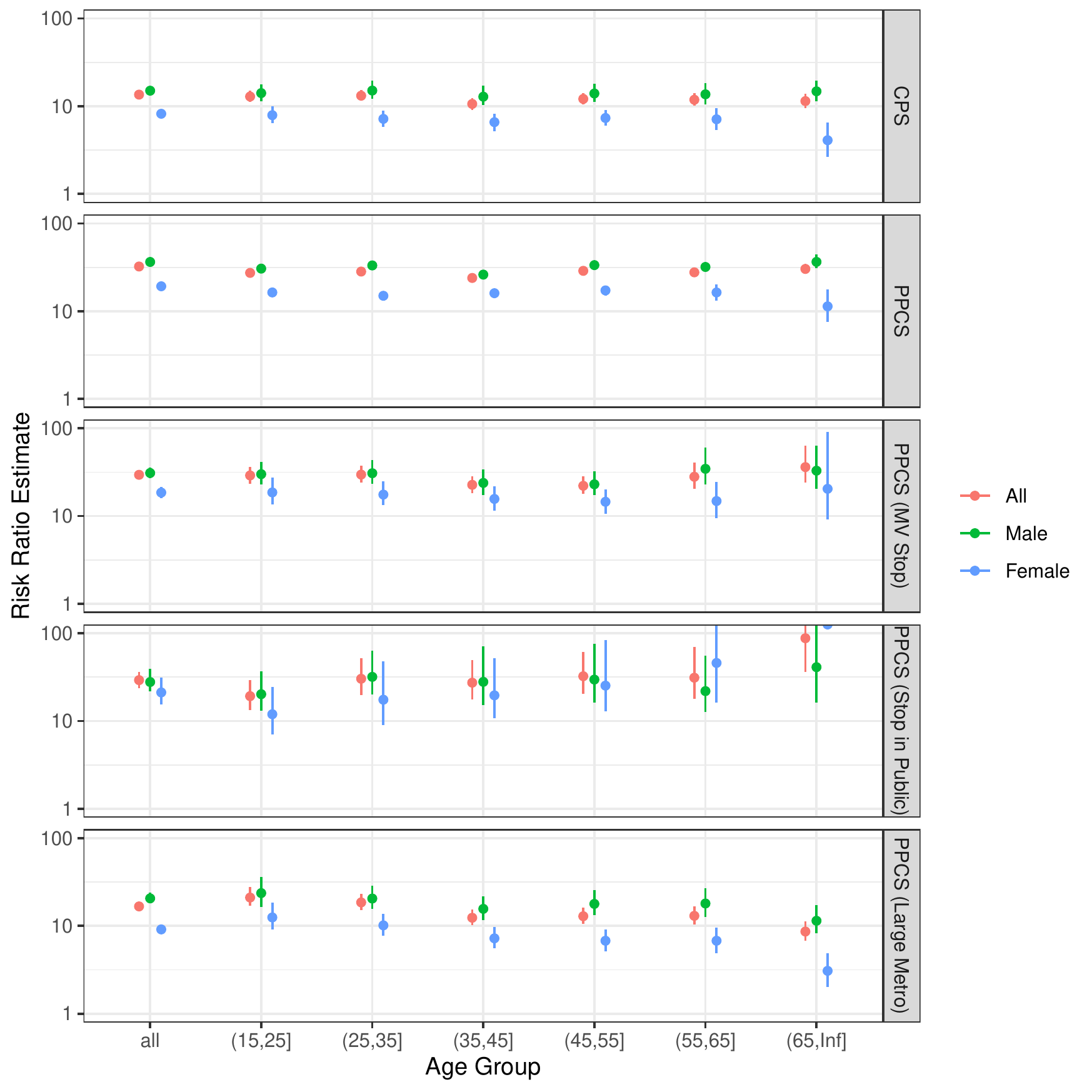}
  \caption{Results of the stratified analysis of the NYPD
    Stop-and-Frisk dataset by age and gender. The estimated risk ratio
    is truncated at 100.}
  \label{fig:stratified}
\end{figure}

\end{document}